\documentclass[conference]{IEEEtran}
\IEEEoverridecommandlockouts

\usepackage{graphicx,cite,amssymb,amsmath,psfrag,subfigure}
\usepackage{graphicx}
\usepackage{amssymb}
\usepackage{amsmath}
\usepackage{cite}
\usepackage{subfigure}
\usepackage{mathrsfs}
\usepackage[displaymath,mathlines]{lineno}
\usepackage{color}
\usepackage{tabulary}
\usepackage{multirow}
\usepackage{algpseudocode}
\usepackage{algorithm,algpseudocode}
\usepackage{algorithmicx}
\usepackage{pbox}
\usepackage{multicol}
\usepackage{lipsum}
\usepackage{amsthm}
\usepackage{epstopdf}
\newtheorem{theorem}{Theorem}
\newtheorem{lemma}{Lemma}

\newtheorem{remark}{Remark}

\ifCLASSINFOpdf
\else
\fi
\begin{document}
\title{Multi-Cell Massive MIMO Performance with Double Scattering Channels}

\author{\IEEEauthorblockN{ Trinh Van Chien, Emil Bj\"{o}rnson, and Erik G. Larsson}
\IEEEauthorblockA{Department of Electrical
    Engineering (ISY), Link\"{o}ping University, SE-581 83 Link\"{o}ping, Sweden\\
\{trinh.van.chien, emil.bjornson, erik.g.larsson\}@liu.se}
\thanks{This paper was supported by the European Union's Horizon 2020 research and innovation programme under grant agreement No 641985 (5Gwireless). It was also supported by ELLIIT and CENIIT.}
}


%Telephone: (800) 555--1212, Fax: (888) 555--1212}


\maketitle

\begin{abstract}
This paper investigates the spectral efficiency (SE) of multi-cell Massive Multi-Input Multi-Output (MIMO) using different channel models. Prior works have derived closed-form SE bounds and approximations for Gaussian distributed channels, while we consider the \emph{double scattering model}---a prime example of a non-Gaussian channel for which it is intractable to obtain closed form SE expressions. The channels are estimated using limited resources, which gives rise to pilot contamination, and the estimates are used for linear detection and to compute the SE numerically. Analytical and numerical examples are used to describe the key behaviors of the double scattering models, which differ from conventional Massive MIMO models. Finally, we provide multi-cell simulation results that compare the double scattering model with uncorrelated Rayleigh fading and explain under what conditions we can expect to achieve similar SEs.
\end{abstract}


% \fi
\IEEEpeerreviewmaketitle

\section{Introduction}
\vspace{-0.10cm}
Massive MIMO has emerged as a key concept for future wireless access due to its potentials to increase both energy efficiency and spectral efficiency (SE) \cite{Marzetta2010a, Ngo2013a}. The communication theoretic foundation has been established under the assumption that the channels between transmitters and receivers are exposed to rich scattering that   can be modeled by uncorrelated Rayleigh fading. The main merit of this channel model is that system performance can be studied in detail using closed form lower bounds on the ergodic capacity, which are obtained by computing the moments of Gaussian distributions \cite{Bjornson2016b}. The channel hardening and favorable propagation properties of Massive MIMO can then be proved analytically \cite{Ngo2014a}. A few works \cite{Huh2012a,Hoydis2013a} have also considered correlated Rayleigh fading. However, real propagation channels are likely to be non-Gaussian distributed and it is thus important to be able to evaluate in SE also for such practical channels, although it is not easy to obtain closed-form expressions.  

The authors in \cite{Ngo2016a} proved mathematically that the channel hardening property does not hold for keyhole channels and SE is degraded by the rank deficiency of the channel correlation matrix, which is characterized by the number of keyholes. One of the most versatile stochastic channel models is the double scattering model \cite{Gesbert2002a}, which utilizes geometry of the propagation environment to model spatial fading correlation, rank-deficiency, limited scattering, etc. Thanks to its flexibility, this model has gained lots of interest for MIMO communications; see \cite{Shin2003a, Hoydis2011e} and references therein. This model has recently been extended for use in Massive MIMO \cite{Wu2014a,Li2015b}. The paper \cite{Wu2014a} considered non-stationary channels where the scattering cluster evolution is modeled by birth-death processes. The paper \cite{Li2015b} studied an upper bound on the single-cell Massive MIMO performance for non-wide sense stationary channels, based on the concepts of partially and wholly visible scattering clusters. However, none of these papers considered practical multi-cell Massive MIMO systems with linear processing and pilot contamination.

In this paper, we describe how to apply the double scattering channel model from \cite{Gesbert2002a} to multi-cell Massive MIMO systems with linear detection. A general uplink SE expression for arbitrary non-Gaussian channels and linear detection is derived. The BS obtains channel state information (CSI) from uplink pilot transmissions, that are exposed to pilot contamination, and applies linear minimum mean squared error (LMMSE) estimation techniques that do not require the exact channel statistics. The key behaviors of the double scattering model (e.g., spatial correlation and favorable propagation) and its impact on the SE are analyzed and illustrated by numerical examples. We consider the linear techniques maximum ratio (MR), zero forcing (ZF), and minimum mean squared error (MMSE) for signal detection. In particular, we compare the results with the uncorrelated Rayleigh fading model.

\textit{Notation:}  We use the upper-case bold face letters for matrices while lower-case bold face ones are used for vectors. $\mathbf{I}_M$, $\mathbf{I}_K$, and $\mathbf{I}_S$ are respectively the identity matrix of size $M \times M$,  $K \times K$ and $S \times S$. The operator $\mathbb{E} \left\{ \cdot \right\}$ represents  the expectation of a random variable. The notation $ \| \cdot \| $ stands for the Euclidean norm. The regular and Hermitian transposes are denoted by $(\cdot)^T$ and  $(\cdot)^H$, respectively. Finally, $\mathcal{CN}(\cdot,\cdot)$ represents the circularly symmetric complex Gaussian distribution.

%==================================================================
%=================================================================
\section{Massive MIMO Cellular Systems} \label{Section:System-Model and Achievable Performance}
\vspace{-0.10cm}
 A Massive MIMO system with $L$ cells is studied. Each cell comprises a BS with $M$ antennas and serves $K$ single-antenna users in the same time-frequency resource. The network operation is divided into channel coherence intervals of length $\tau_c$ symbols. In each coherence interval, the channel $\mathbf{h}_{i,t}^l \in \mathbb{C}^{M}$ between user $t$ in cell $i$ and BS $l$ is constant and flat-fading (for any $i,t,l$), while independent channel realizations are assumed between coherence intervals. Moreover, $\tau_p$ symbols are used for pilot signaling and the remaining $\tau_c - \tau_p$ symbols are dedicated to data transmission. Due to space limitations, we only focus on uplink (UL) transmission in this paper, but the results can be extended to the downlink using a time division duplex (TDD) protocol that exploits channel reciprocity.

\subsection{Uplink Channel Estimation}
\vspace{-0.10cm}
Each BS needs to know the channels to its users to make efficient use of its $M$ antennas. For channel estimation purposes, we assume that all users simultaneously transmit pilot sequences of length $\tau_p$, with $\tau_p = f K$, where the positive integer $f$ denotes the pilot reuse factor. Orthogonal pilots are used within each cell, while users in different cells may use the same pilot if the length of the pilot sequences is less than the total number of users in the network, i.e., $\tau_p < LK$. Because $f < L$, the pilot sequences are reused across the cells in the network. The consequences of such pilot contamination are described later. The received pilot signal $\mathbf{Y}_l \in \mathbb{C}^{M \times \tau_p}$ at BS $l$ is expressed
\begin{equation} \label{eq: UL-Recieved-Pilot-Matrix}
\mathbf{Y}_l = \sum_{i=1}^L \mathbf{H}_i^l \mathbf{P}_i^{1/2} \pmb{ \Phi }_i^H + \mathbf{N}_l,
\end{equation}
where the channel matrix  $\mathbf{H}_i^l = [\mathbf{h}_{i,1}^l, \ldots,  \mathbf{h}_{i,K}^l ] \in \mathbb{C}^{M \times K}$ has the column vectors $\mathbf{h}_{i,k}^l$  each of which denotes the channel between user $k$ in cell $i$ and BS $l$, for $k=1,\ldots,K$, $i =1,\ldots, L$, and $l =1,\ldots, L$.  The orthogonality of pilot sequences in a cell implies that the $\tau_p \times K$ pilot matrix used in cell $i$, $\pmb{\Phi}_i =  [\pmb{\phi}_{i,1}, \ldots , \pmb{\phi}_{i,K}]$, satisfies $\pmb{\Phi}_i^H \pmb{\Phi}_i = \tau_p \mathbf{I}_K$. The pilots are divided into $f$ distinct groups and each cell belongs to one such group. If cell $i$ and cell $l$ use the same pilot sequences, we also have $\pmb{\Phi}_i^H \pmb{\Phi}_l = \tau_p \mathbf{I}_K$. Otherwise, $\pmb{\Phi}_i^H \pmb{\Phi}_l = \mathbf{0}$ meaning that these cells use different orthogonal pilot sequences. Let us denote by $\mathcal{P}_{l} \subset \{1,\ldots, L\}$ the indices of the cells employing the same pilot sequence as cell $l$. Additionally, we let $p_{i,k}$ denote the transmit power of user $k$ in cell $i$ and define the diagonal power matrix $\mathbf{P}_i = \mathrm{diag} (p_{i,1}, \ldots, p_{i,K}) \in \mathbb{C}^{K \times K}$. Finally,  $\mathbf{N}_l \in \mathbb{C}^{M \times \tau_p}$ is a noise matrix with independent entries having the distribution $\mathcal{CN} (0, \sigma_{\mathrm{UL}}^2 )$. 

We evaluate and compare the performance of Massive MIMO using different channel models, thus no particular channel model is assumed at this point. However, for sake of simplicity,  we only cover non line-of-sight channels in this paper, i.e., they have zero mean. Each BS is assumed to know the first and second order moments of the channels from all users, while the statistical distribution is unknown (as is typically the case in practice). Based on the received pilot signal in \eqref{eq: UL-Recieved-Pilot-Matrix}, BS $l$ can then apply LMMSE estimation \cite{Kay1993a} to obtain an estimate $\hat{ \mathbf{h} }_{l,k}^l$ of $\mathbf{h}_{l,k}^l$.
A sufficient statistic for estimating $\mathbf{h}_{l,k}^l$ is
\begin{equation}
\begin{split} \label{eq:UL-Recieved-Pilot-Matrix-processed}
\mathbf{Y}_l \pmb{\phi}_{l,k}   &= \mathbf{y}_{l,k}^l = \sum_{i=1}^L \mathbf{H}_i^l \mathbf{P}_i^{1/2} \pmb{ \Phi }_i^H \pmb{\phi}_{l,k}  + \mathbf{N}_l \pmb{\phi}_{l,k} \\
&= \sum_{ i \in \mathcal{P}_l } \sqrt{p_{i,k}} \tau_p \mathbf{h}_{i,k}^l + \tilde{ \mathbf{n} }_{l,k}^l,
\end{split}
\end{equation}
where $\tilde{ \mathbf{n} }_{l,k}^l = \mathbf{N}_l \pmb{\phi}_{l,k} \sim \mathcal{CN}( \mathbf{0}, \tau_p \sigma_{\mathrm{UL}}^2 \mathbf{I}_M )$. 
Note that \eqref{eq:UL-Recieved-Pilot-Matrix-processed} contains a summation $\sum_{ i \in \mathcal{P}_l } \mathbf{h}_{i,k}^l$ of the channels from users that transmitted the same pilot sequence $\pmb{\phi}_{l,k}$.
%The LMMSE channel estimate is provided in Lemma \ref{Lemma:EstimatedChannel}.

\begin{lemma} \label{Lemma:EstimatedChannel}
The LMMSE estimate of the channel between user $k$ in cell $l$ and BS $l$,  $\hat{\mathbf{h}}_{l,k}^{l}$, is 
\begin{equation} \label{eq:ChannelEstimate}
\hat{\mathbf{h}}_{l,k}^{l} = \mathbf{B}_{l,k}^l \mathbf{y}_{l,k}^l = \mathbf{B}_{l,k}^l \left( \sum_{i \in \mathcal{P}_l } \sqrt{p_{i,k}} \tau_p \mathbf{h}_{i,k}^l + \tilde{ \mathbf{n} }_{l,k} \right),
\end{equation}
where $\mathbf{B}_{l,k}^l =  \mathbf{Cov} \{ 
\mathbf{h}_{l,k}^l, \mathbf{y}_{l,k}^l \} (\mathbf{Cov}\{ 
\mathbf{y}_{l,k}^l, 
\mathbf{y}_{l,k}^l \} )^{-1}$ and the cross-correlation matrix $\mathbf{Cov} \{ 
\mathbf{h}_{l,k}^l, \mathbf{y}_{l,k}^l \}$ and the auto-correlation matrix $\mathbf{Cov}\{ 
\mathbf{y}_{l,k}^l, 
\mathbf{y}_{l,k}^l \}$ are defined as
\begin{align}
\mathbf{Cov} \{ 
\mathbf{h}_{l,k}^l, \mathbf{y}_{l,k}^l \} &=  \sqrt{p_{l,k}} \tau_p \mathbb{E} \{ \mathbf{h}_{l,k}^l (\mathbf{h}_{l,k}^l)^H \}\\
\mathbf{Cov}\{ 
\mathbf{y}_{l,k}^l, 
\mathbf{y}_{l,k}^l \} &= \tau_p^2  \sum_{i \in \mathcal{P}_l }p_{i,k} \mathbb{E} \{ \mathbf{h}_{i,k}^l (\mathbf{h}_{i,k}^l)^H \} +  \tau_p \sigma_{\mathrm{UL}}^2  \mathbf{I}_{M}.
\end{align}
\end{lemma}
\begin{proof}
	The proof follows directly from the LMMSE techniques in \cite{Kay1993a}.
\end{proof}

We stress that the estimator in Lemma \ref{Lemma:EstimatedChannel} is very general since it applies to any channel with zero mean; in particular, the channel does not have to be Gaussian distributed, as otherwise assumed in most of the Massive MIMO literature. This implies that $\hat{\mathbf{h}}_{l,k}^{l}$ and the estimation error $\mathbf{h}_{l,k}^{l}- \hat{\mathbf{h}}_{l,k}^{l}$ are uncorrelated, but generally not independent. The LMMSE estimator above includes the pilot contamination that often occurs in multi-cell systems. 
%In the special case when all users have mutual orthogonal pilots, the BSs can obtain better estimated CSI as shown in Corollary \ref{Remark-OrthogonalEstimate}.
%\begin{corollary} \label{Remark-OrthogonalEstimate}
%If all users in the cells utilize mutually orthogonal pilot sequences, i.e. $\tau_p \geq LK$, then the channel estimate in \eqref{eq:ChannelEstimate} is simplified to
%\begin{equation}
%\hat{\mathbf{h}}_{l,k}^{l} = \mathbf{B}_{l,k}^l \left( \sqrt{p_{l,k}} \tau_p \mathbf{h}_{l,k}^l + \tilde{ \mathbf{n} }_{l,k} \right),
%\end{equation}
%where the deterministic matrix, $\mathbf{B}_{l,k}^l$ is defined as
%\begin{equation*}
%\sqrt{p_{l,k}} \mathbb{E} \{ \mathbf{h}_{l,k}^l (\mathbf{h}_{l,k}^l)^H \} \left( \tau_p p_{l,k} \mathbb{E} \{ \mathbf{h}_{l,k}^l (\mathbf{h}_{l,k}^l)^H \} +  \sigma_{\mathrm{UL}}^2 \mathbf{I}_{M} \right)^{-1} .
%\end{equation*}
%\end{corollary}
The LMMSE channel estimates in cell $l$ are written in the compact form as $\widehat{\mathbf{H}}_{l}^l = [ \hat{\mathbf{h}}_{l,1}^l, \ldots, \hat{\mathbf{h}}_{l,K}^l]^T \in \mathbb{C}^{M \times K}$. 
These will be used for UL detection.

\subsection{Uplink Data Transmission Model}
\vspace{-0.10cm}
 We assume that the arbitrary user $t$ in cell $i$ transmits the signal $x_{i,t} \in \mathbb{C}$, having unit power $\mathbb{E} \{ |x_{i,t}|^2\} = 1$.  For the UL data transmission, the received signal at BS $l$ is modeled as
\begin{equation} \label{eq: UL-Signal}
\begin{split}
 \mathbf{y}_l = \sum_{i=1}^L \sum_{t=1}^K \sqrt{p_{i,t}}\mathbf{h}_{i,t}^l x_{i,t} + \mathbf{n}_l,
 \end{split}
\end{equation}
where  $\mathbf{n}_l \sim \mathcal{CN}(0, \sigma_{\mathrm{UL}}^2 \mathbf{I}_M )$ is additive noise. By using the detection vector $\mathbf{v}_{l,k} \in \mathbb{C}^M$, user $k$ in cell $l$ can detect the transmitted signal $x_{l,k}$ as
\begin{equation}\label{eq: UL-SignalCombining}
\begin{split}
\mathbf{v}_{l,k}^H \mathbf{y}_{l} &=  \sqrt{ p_{l,k} } \mathbf{v}_{l,k}^H \mathbf{h}_{l,k}^l  x_{l,k} + \sum\limits_{\substack{t = 1\\t \neq k}}^{K} \sqrt{ p_{l,t} } \mathbf{v}_{l,k}^H \mathbf{h}_{l,t}^l x_{l,t} \\
& \; \; \; + \sum\limits_{ \substack{i =1 \\ i \neq l} }^{L} \sum\limits_{ t =1 }^{K} \sqrt{ \rho_{i,t} } \mathbf{v}_{l,k}^H \mathbf{h}_{i,t}^l x_{i,t} +  \mathbf{v}_{l,k}^H \mathbf{n}_l. 
\end{split}
\end{equation}
The first term in \eqref{eq: UL-SignalCombining} is the desired signal that BS $l$ receives from user $k$. The second term is intra-cell interference that comes from the other users in cell $l$. The third term is inter-cell interference while the last one is additive noise. %The inter-user interference, particularly from users within the own cell and at the edge of neighboring cells often provides the largest amount of interference.
From \eqref{eq: UL-SignalCombining}, a lower bound on the ergodic capacity of user $k$ in cell $l$ is given in Theorem \ref{Theorem-Lower-Bound-Rate}. 

\begin{theorem} \label{Theorem-Lower-Bound-Rate}
 A lower bound on the UL ergodic capacity of user $k$  in cell $l$ is given by
\begin{equation} \label{eq: Sum-Rate-k}
R_{l,k}^{\mathrm{UL}} =  \left( 1 - \frac{\tau_p}{\tau_c} \right) \log_2 \left(1 + \mathrm{SINR}_{l,k}^{\mathrm{UL}} \right) \quad \textrm{[bit/s/Hz]},
\end{equation}
where the effective signal-to-noise-and-interference (SINR) value, $\mathrm{SINR}_{l,k}^{\mathrm{UL}}$, is
\begin{equation} \label{SINR:UL}
\fontsize{9.5pt}{9.5}{\frac{p_{l,k} | \mathbb{E} \{ \mathbf{v}_{l,k}^H \mathbf{h}_{l,k}^l \}|^2}{ \sum\limits_{i=1}^L \sum\limits_{t=1}^K p_{i,t} \mathbb{E} \{ | \mathbf{v}_{l,k}^H \mathbf{h}_{i,t}^l |^2 \} - p_{l,k} | \mathbb{E} \{ \mathbf{v}_{l,k}^H \mathbf{h}_{l,k}^l \}|^2 + \sigma_{\mathrm{UL}}^2  \mathbb{E} \{ \| \mathbf{v}_{l,k} \|^2 \} }.}
\end{equation}
\end{theorem}
\begin{proof}
We obtain this bound on the ergodic capacity by assuming Gaussian signaling and applying a series of lower bounds that reduce the mutual information between the transmitted and received signals. A detailed proof that applies to channels of arbitrary distribution is available in \cite{Chien2016b}.
\end{proof}

The capacity bound in Theorem \ref{Theorem-Lower-Bound-Rate} is used as SE expression in this paper. While the expression resembles the SE expression derived for Rayleigh fading \cite{Bjornson2016a}, we stress that our expression is valid for a much wider range of channel models. The expectations in the expression can be computed numerically for any channel distribution and choice of linear detection vector $\mathbf{v}_{l,k}^l, \forall l,k$. We define the detection matrix at BS $l$ as $\mathbf{V}_{l} = [\mathbf{v}_{l,1}, \ldots, \mathbf{v}_{l,K}] \in \mathbb{C}^{M \times K}$ and consider the three most common detection techniques in Massive MIMO \cite{Ngo2013a}:
\begin{equation}
\mathbf{V}_{l} = \begin{cases}
\widehat{\mathbf{H}}_l^l & \mbox{for MR}, \\
\widehat{\mathbf{H}}_l^l \left( (\widehat{\mathbf{H}}_l^l)^H \widehat{\mathbf{H}}_l^l \right)^{-1} & \mbox{for ZF}, \\
\widehat{\mathbf{H}}_l^l \left( (\widehat{\mathbf{H}}_l^l)^H \widehat{\mathbf{H}}_l^l + \mathbf{P}_{l}^{-1}  \right)^{-1} & \mbox{for MMSE}. \\
\end{cases}
\end{equation}

%====Section: Power Optimization for Multi-cell Massive MIMO =======
\section{Double Scattering Model for Massive MIMO} \label{Section:Channel Modeling}
\vspace{-0.1cm}
This section describes how to adapt the double scattering model, initially proposed for point-to-point MIMO \cite{Gesbert2002a}, to multi-cell Massive MIMO systems. To motivate the use of that model, we first briefly review the conventional uncorrelated Rayleigh fading and its modeling assumptions.

\vspace*{-0.1cm}
\subsection{Uncorrelated Rayleigh Channels}\label{SubSec:UncorrelatedRayleigh}
\vspace{-0.2cm}
\begin{figure}[t]
	\centering
	\includegraphics[trim=0.5cm 1.0cm 0.5cm 1.0cm, clip=true, width=3in]{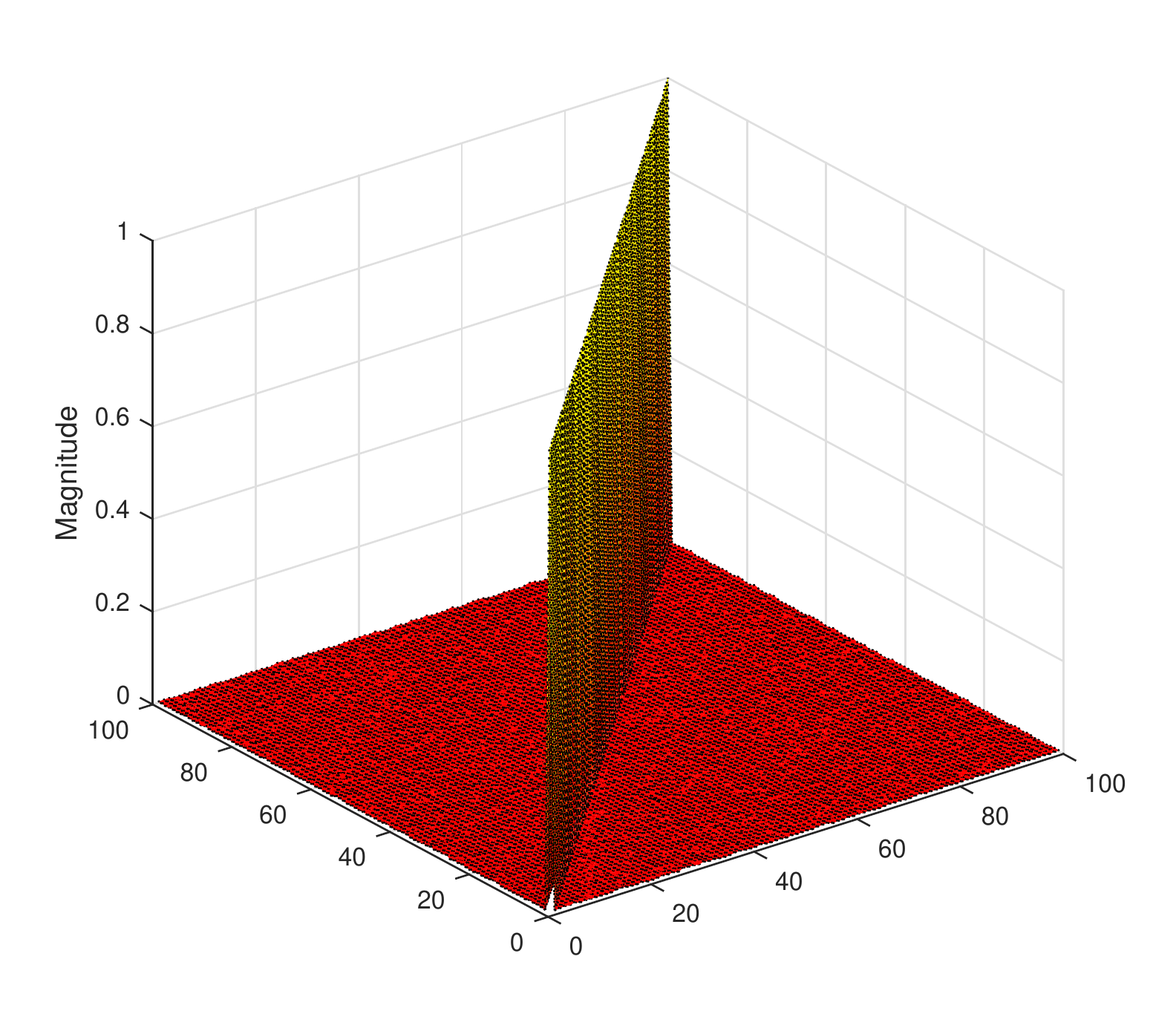} \vspace{-0.20cm}
	\caption{ The magnitude of $\mathbb{E} \{ \mathbf{h}_{l,k}^l (\mathbf{h}_{l,k}^l)^H \} / \beta_{l,k}^l$  of the uncorrelated Rayleigh fading channels. Here, $l=1, k =1, M = 100$, and the result is averaged over the $100,000$ realizations of the small-scale fading. }
	\label{fig:Rayleigh}
	\vspace{-0.20cm}
\end{figure}

In the uncorrelated Rayleigh fading model, the channel between user $k$ in cell $l$ and BS $l$ is distributed as
\begin{equation} \label{eq:UncorrelatedRayleigh}
 \mathbf{h}_{l,k}^l \sim \mathcal{CN}(\mathbf{0}, \beta_{l,k}^l \mathbf{I}_M).
 \end{equation}
 where $\beta_{l,k}^l$ represents the large-scale fading that describes macroscopic attenuation and shadowing. 
% By letting $\mathbf{D}_{l}^l = \textrm{diag} ( \sqrt{\beta_{l,1}^l}, \cdots, \sqrt{\beta_{l,K}^l} ) \in \mathbb{C}^{K \times K}$ and $\mathbf{G}_{l}^l \in \mathbb{C}^{K \times K}$ denote the small-scale fading matrix where its entries flows $\mathcal{CN}(0,1)$, the channel matrix between the users in cell $l$ and BS $l$, $\mathbf{H}_l^l$, is expressed as
% \begin{equation}
%\mathbf{H}_l^l = \mathbf{G}_{l}^l \mathbf{D}_{l}^l
% \end{equation}
This model is popular in the Massive MIMO literature since a lower bound on the ergodic capacity can be computed in closed forms for MR and ZF \cite{Ngo2013a,Chien2016b}, utilizing the fact that 
moments of Gaussian distributions are analytically tractable.  Moreover, these channels manifest two fundamental properties in Massive MIMO: channel hardening and favorable propagation \cite{Ngo2014a}. 

Fig.~\ref{fig:Rayleigh} illustrates the meaning of uncorrelated fading by plotting the magnitude of $\mathbb{E} \{ \mathbf{h}_{l,k}^l (\mathbf{h}_{l,k}^l)^H \} / \beta_{l,k}^l$, based on $100,000$ Monte-Carlo realizations of the small-scale fading. The main diagonal becomes one and off-diagonal elements are all zero, which implies that the signal from the user has no dominant directivity.
There are two main issues with the uncorrelated Rayleigh fading model in Massive MIMO. First, spatially correlated fading has been observed in practical measurements \cite{Gao2015c} and it is rather intuitive that the main energy will arrive from around the spatial direction to the user. Second, uncorrelated Rayleigh fading is motivated by having a large number of scattering objects as compared to the number of antennas, which might not be the case when $M$ is large.

\subsection{Double Scattering Channels}
\vspace{-0.10cm}
\begin{figure}[t]
	\vspace{-0.20cm}
	\centering
	\includegraphics[trim=4cm 10.5cm 13cm 2.3cm, clip=true, width=2.5in]{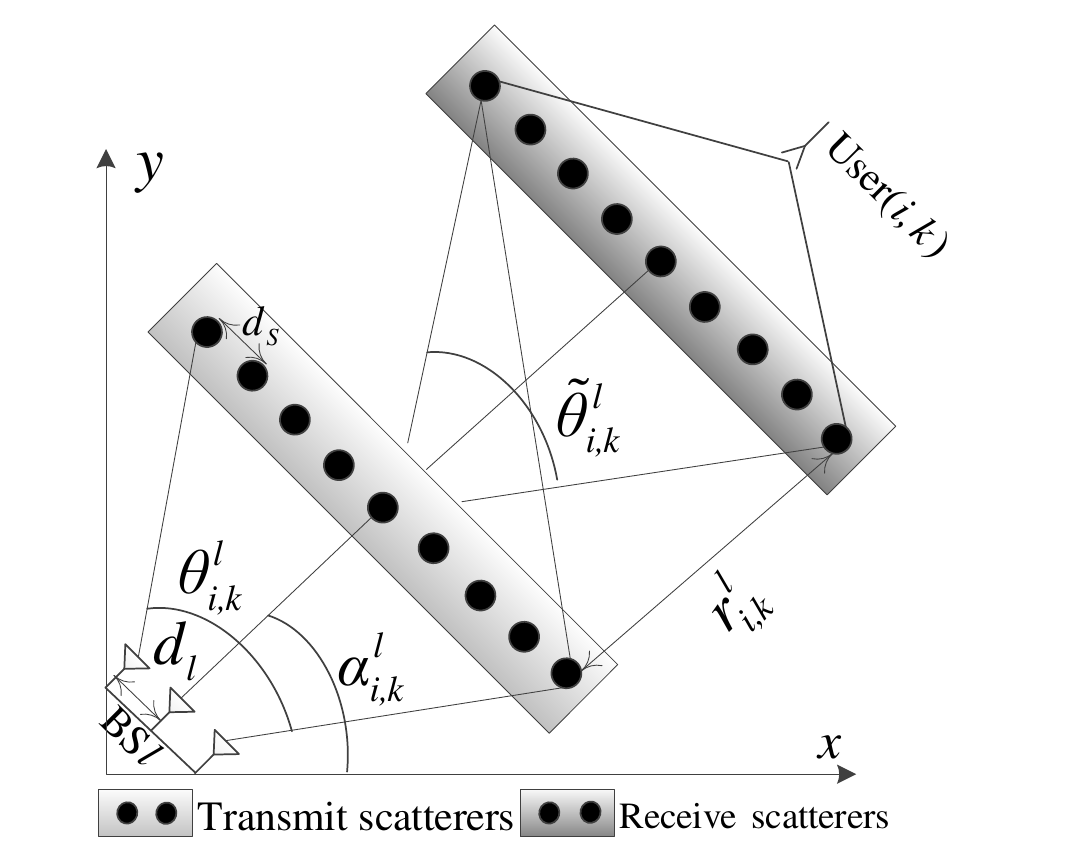} \vspace{-0.15cm}
	\caption{The geometric model of the double scattering channel between user $k$ in cell $i$ and BS $l$. }
	\label{fig:Doublescatteringstructure}
	\vspace{-0.40cm}
\end{figure} 

A main contribution of this paper is to apply the double scattering channel model to multi-cell Massive MIMO with linear detection and analyze the corresponding results. This channel model is a multivariate function  of the scattering distribution at both the BS and the user, the distance between them, and the spacing between the antennas. This model provides non-Gaussian channels with spatial correlation, and thus is a way to address the aforementioned issues of the uncorrelated Rayleigh fading model.

A schematic of the double scatting channels between an arbitrary user $k$ in cell $i$ and BS $l$ is depicted in Fig.~\ref{fig:Doublescatteringstructure} with several geometric parameters. %\footnote{ We use the uniform distributions for the scatterers since it makes simplicity in defining the correlation matrix between the scatterers. The cases of other distributions are left for future works.}
The transmit and receive scattering clusters are respectively located at the BS and the user sides and the distance between them is $r_{i,k}^l$. There is a lack of scattering in between the two clusters which gives rise to spatial correlation and non-Gaussian fading. Meanwhile the distance between two scatterers is $d_S$. The angular spread between the transmit and receive scatterers is denoted by $\tilde{\theta}_{i,k}^l$. The distance between two BS antennas is $d_{l}$ and the angular spread is $\theta_{i,k}^l$. The azimuth angle between the BS antenna array and the transmit clusters is $\alpha_{i,k}^l$. Mathematically, the channel of user $k$ in cell $i$ and BS $l$ is formulated as
\begin{equation} \label{eq:DoubleScatteringChannel}
\mathbf{h}_{i,k}^l = \sqrt{\frac{\beta_{i,k}^l}{S}} (\mathbf{R}_{i,k}^{l})^{1/2} \mathbf{G}_{i,k}^{l} (\widetilde{\mathbf{R}}_{i,k}^l)^{1/2} \tilde{\mathbf{g}}_{i,k}^l,
\end{equation}
where $S$ is the number of scatterers at the BS side and at the user side. $\mathbf{G}_{i,k}^{l} \in  \mathbb{C}^{M \times S}$ describes the small-scale fading between the BS and its scattering cluster and $\tilde{\mathbf{g}}_{i,k}^l \in  \mathbb{C}^{S}$ describes the small-scale fading between the users and its scattering cluster. Similar to \cite{Gesbert2002a}, we assume the number of scatterers are sufficient large (but not necessarily larger than the number of BS antennas) to ensure that $\mathbf{G}_{i,k}^{l} \in  \mathbb{C}^{M \times S}$ and $\tilde{\mathbf{g}}_{i,k}^l \in  \mathbb{C}^{S}$  have entries following Gaussian distribution, $\mathcal{CN} (0,1)$. The correlation matrix $\mathbf{R}_{i,k}^{l} \in \mathbf{C}^{M \times M}$
between the BS antennas and the $S$ transmit scatterers has its $(m,m')$th element, $\forall m,m' = 1, \ldots, M,$  computed as \cite{Gesbert2002a}
\begin{equation}
\begin{split}
&[ \mathbf{R}_{i,k}^l ]_{m,m'} =\frac{1}{S} \times \\
&  \sum_{ n = (1-S)/2 }^{(S-1)/2} \exp \left\{ - 2 \pi j (m-m') d_l \cos \left( \frac{\pi}{2} + \alpha_{i,k}^l + \theta_{i,k,n}^l \right) \right\},
\end{split}
\end{equation}
where $ \theta_{i,k,n}^l$  denotes the angle spread between the $n$th scatterer and the BS antenna array. It is computed as \cite{Ertel1998a}
\begin{equation}
\theta_{i,k,n}^l = \frac{n \theta_{i,k}^l }{ S- 1}, \forall n =  \frac{1-S}{2}, \ldots, \frac{S-1}{2}.
\end{equation}
The correlation matrix $\mathbf{R}_{i,k}^l$ models the influence of the scattering distribution and spacing and beamforming of antennas. When the BS antenna separation $d_l$ becomes large the correlation matrix becomes uncorrelated \cite{Gesbert2002a}, i.e., $\mathbf{R}_{i,k }^l \rightarrow \mathbf{I}_M $, but in general it is not an identity matrix.

The correlation matrix between the transmit and receive scatterers is computed by considering them as virtual antenna arrays such that the angle spread $\tilde{\theta}_{i,k}^l$ is calculated as \cite{Gesbert2002a}
\begin{equation}
\tan \left( \frac{\tilde{\theta}_{i,k}^l}{2} \right) = \frac{d_S (S-1)}{ 2 r_{i,k}^l },
\end{equation}
and then the $(m,m')$th element, $\forall m,m' = 1, \ldots, S$, of the correlation matrix $\widetilde{\mathbf{R}}_{i,k }^l$ is computed as
\begin{equation} \label{eq:ScatteringCorrelationMatrix}
\begin{split}
&[ \widetilde{\mathbf{R}}_{i,k }^l ]_{m,m'} =\frac{1}{S} \times \\
&  \sum_{ n = (1-S)/2 }^{(S-1)/2} \exp \left\{ - 2 \pi j (m-m') d_{S} \cos \left( \frac{\pi}{2} + \alpha_{i,k}^l + \tilde{\theta}_{i,k,n}^l \right) \right\},
\end{split}
\end{equation}
where the angle $\tilde{\theta}_{i,k,n}^l$ between the scatterers is defined as \cite{Ertel1998a}
\begin{equation}
\tilde{\theta}_{i,k,n}^l = \frac{n \tilde{\theta}_{i,k}^l }{ S- 1}, \forall n =  \frac{1-S}{2}, \ldots, \frac{S-1}{2}.
\end{equation}
\begin{figure}[t]
	\centering
	\includegraphics[trim=0.5cm 0.7cm 0.5cm 1cm, clip=true, width=2.9in]{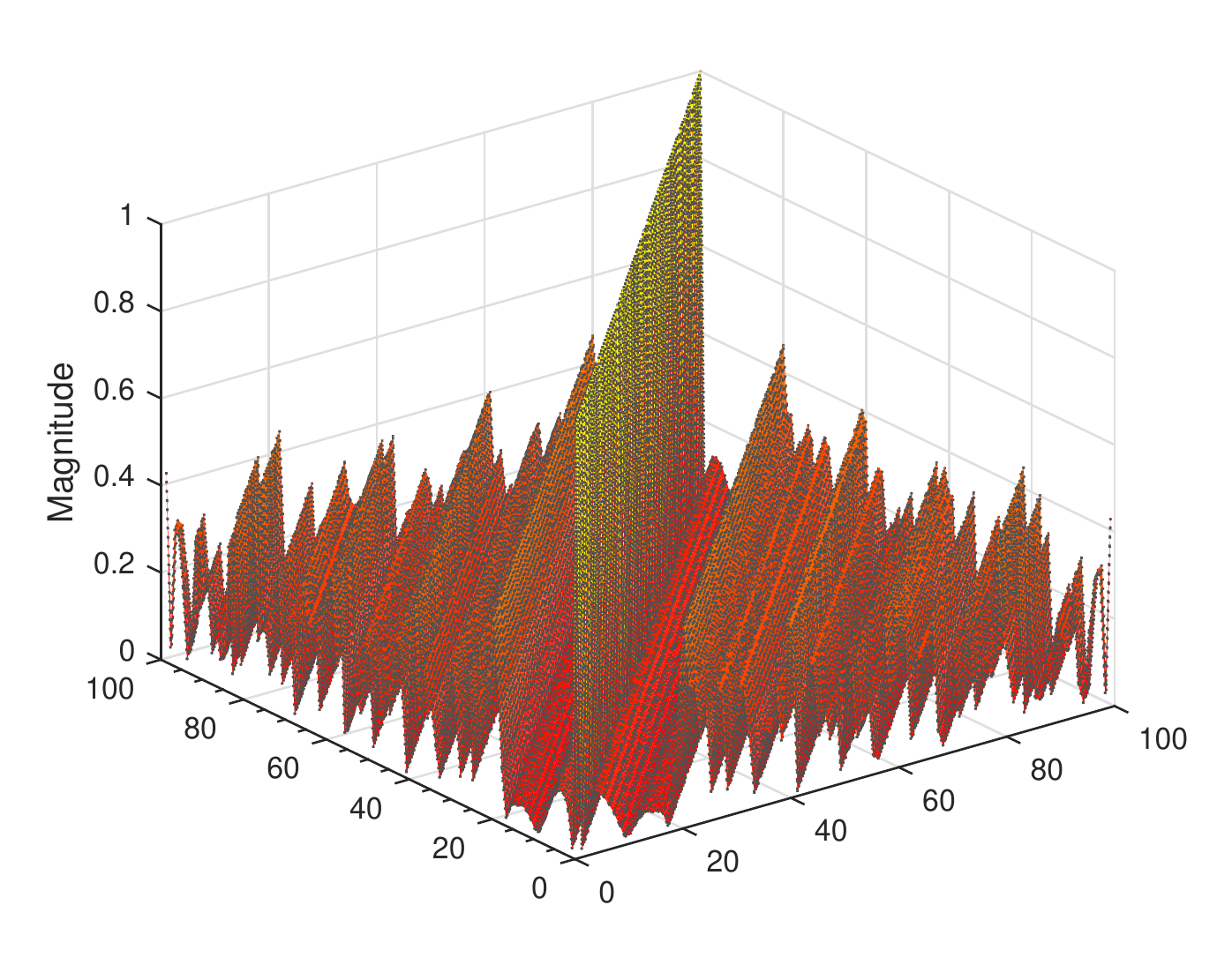} \vspace{-0.2cm}
	\caption{The correlation matrix $\mathbb{E} \{ \mathbf{h}_{l,k}^l (\mathbf{h}_{l,k}^l)^H \} / \beta_{l,k}^l $ of the double scattering channels
		is computed using $100,000$ realizations of the small-scale fading. Here, $l=1, k=1, M = 100, S= 21, d_l = 0.5, \theta_{l,k}^l = 2 \pi/3, \alpha_{l,k}^l = 0$.}
	\label{fig:DoubleScatteringChannelwith21Scatters}
	\vspace*{-0.4cm}
\end{figure}

The double scattering channel model combines three important aspects of Massive MIMO channel propagation, namely the rank-deficiency, the spatial fading correlation, and the signal attenuation by controlling the number of scatterers, the correlation matrices, and the large-scale fading coefficients. Furthermore the model spans scenarios from uncorrelated Rayleigh  to the keyhole channels, as shown in Remark~\ref{remark:asymptotic-cases}.

\begin{remark} \label{remark:asymptotic-cases}
	The uncorrelated Rayleigh fading is obtained from the double scattering model by setting $ \mathbf{R}_{i,k}^l = \mathbf{I}_M$ and $\widetilde{\mathbf{R}}_{ i,k }^l = \mathbf{I}_S$ and then letting $S \rightarrow \infty$, which are assumptions that give perfect antenna and scattering conditions with a high-rank correlation matrix. In contrast, the worst case of rank deficiency, i.e.,  $\widetilde{\mathbf{R}}_{ i,k }^l$ has rank $1$, is characterized by the cases of all entries of $\widetilde{\mathbf{R}}_{ i,k }^l$ being $1$. Moreover, $S =1$ yields the keyhole channels.
\end{remark}

To illustrate these behaviors, Fig.~\ref{fig:DoubleScatteringChannelwith21Scatters} illustrates the normalized correlation matrix $\mathbb{E} \{ \mathbf{h}_{l,k}^l (\mathbf{h}_{l,k}^l)^H \} / \beta_{l,k}^l $ for the double scattering model with $S=21$ scatterers. Different from the uncorrelated Rayleigh fading in Fig.~\ref{fig:Rayleigh}, we observe a clear spatial correlation where the off-diagonal elements are non-zero and follow a non-trivial pattern that describes the propagation environment. The corresponding correlation matrix with $S=81$ is shown in Fig.~\ref{fig:DoubleScatteringChannelwith81Scatterers}. The larger numbers of scatterers make the channel statistics closer to the uncorrelated Rayleigh fading case, but there are still distinct differences. Since the channel distribution is non-Gaussian, the lower bound on the UL ergodic capacity in Theorem \ref{Theorem-Lower-Bound-Rate} cannot be computed in closed-form but requires Monte-Carlo simulations.

The spatial correlation affects the favorable propagation properties, which is illustrated in Fig.~\ref{fig:ChannelHardeningDoubleScattering} by considering the average inner product $\mathbb{E} \{|(\mathbf{h}_{l,k}^l)^H (\mathbf{h}_{l,t}^l)| \} / (M \sqrt{ \beta_{l,k}^l \beta_{l,t}^l}) $ between the user channels. One user has a fixed azimuth angle of $0$ rad, while the angle of the other user is varied between $- \pi$ and $+ \pi$. The spatial correlation creates patterns since users with similar angles are more likely to have similar channel realizations. Interestingly, the variations are smoother as number of scatterers increases and the curves are closer to the uncorrelated Rayleigh fading.

\begin{figure}[t]
	\centering
	\includegraphics[trim=0.5cm 0.8cm 0.5cm 1.0cm, clip=true, width=2.9in]{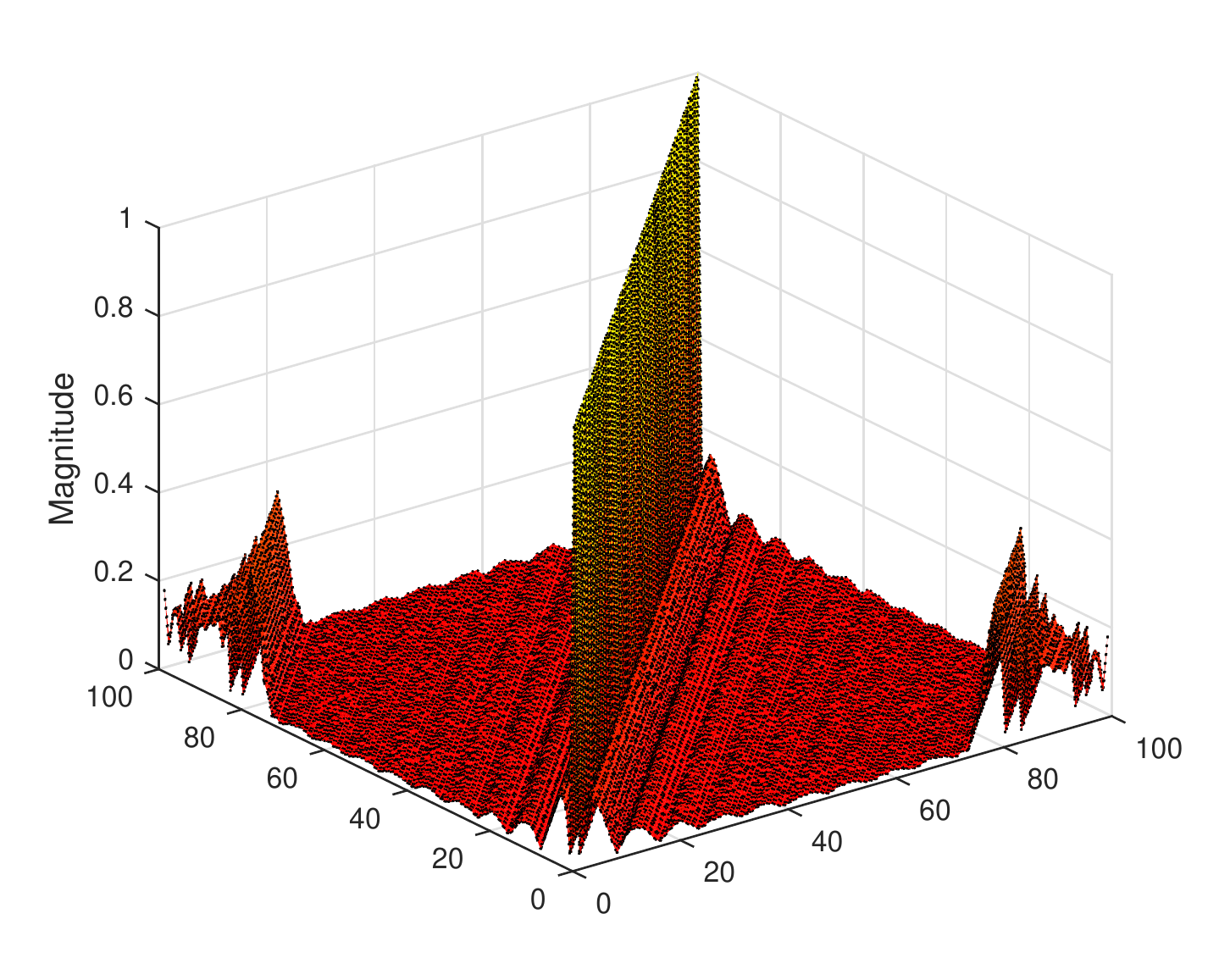} \vspace{-0.20cm}
	\caption{Same scenario as in Fig.~\ref{fig:DoubleScatteringChannelwith21Scatters} but for the number of scatters $S=81$.}
	\label{fig:DoubleScatteringChannelwith81Scatterers}
	\vspace*{-0.4cm}
\end{figure}

\begin{figure}[t]
	\centering
	\includegraphics[trim=4cm 9cm 5.0cm 9cm, clip=true, width=2.7in]{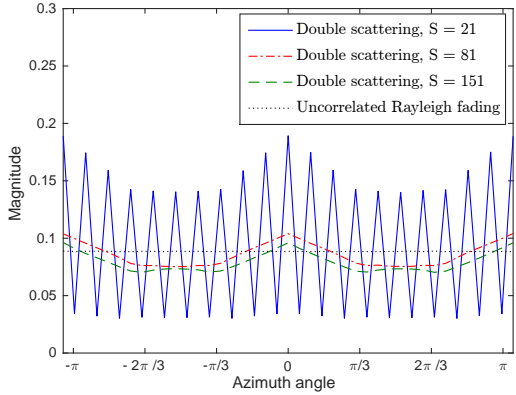} \vspace{-0.20cm}
	\caption{The inner product $\mathbb{E} \{|(\mathbf{h}_{l,k}^l)^H (\mathbf{h}_{l,t}^l)| \} / (M \sqrt{\beta_{l,k}^l \beta_{l,t}^l})$, where the azimuth angle of user $k$ is fixed at 0 rad while the azimuth angle of user $t$ varies in the range of $[- \pi, + \pi]$ rad with a step size of $\pi/20$.}
	\label{fig:ChannelHardeningDoubleScattering}
	\vspace*{-0.4cm}
\end{figure}

%=============================================================
%=============================================================
\section{Numerical Results} \label{Section:Numerical-Results}
\vspace{-0.2cm}
\begin{figure}[t]
	\centering
	\includegraphics[trim=9.5cm 9cm 9.5cm 2.5cm, clip=true, width=2.0in]{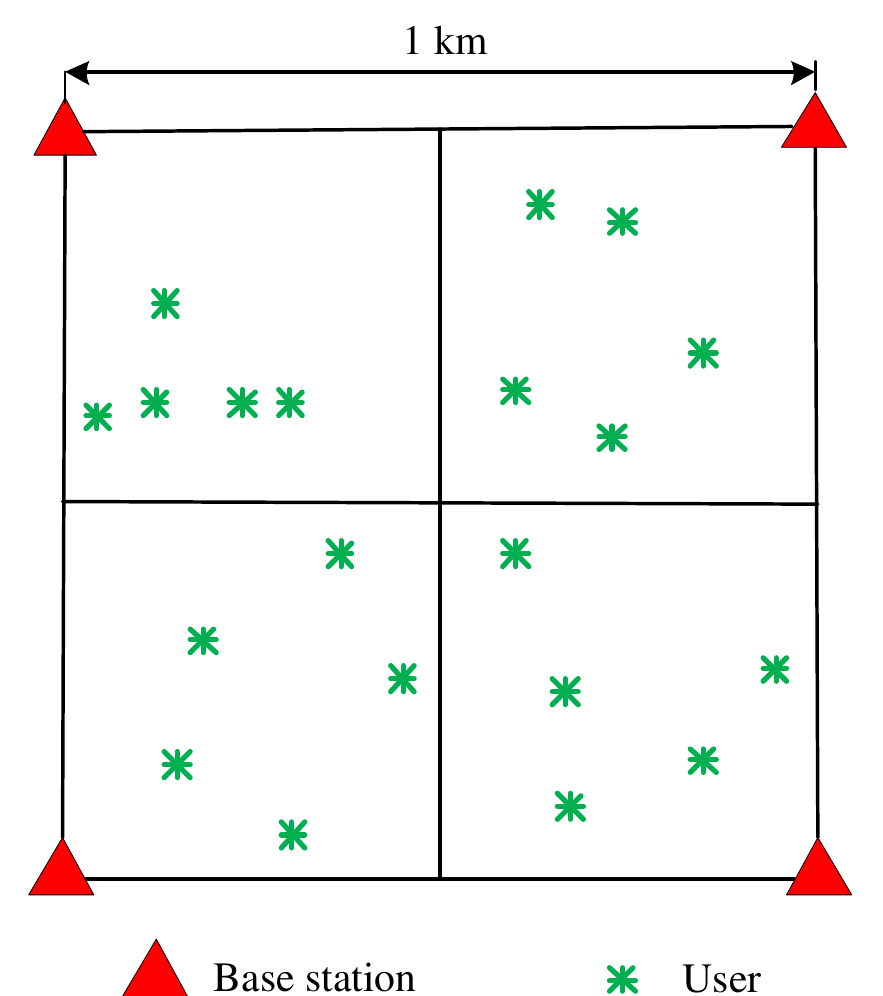} \vspace{-0.20cm}
	\caption{Multi-cell Massive MIMO system considered in the simulations.}
	\label{FigSystemLayout}
	\vspace{-0.30cm}
\end{figure}

In this section we provide numerical results to demonstrate the SE behavior when using the double scattering model presented in Section \ref{Section:Channel Modeling}. We consider a system with $4$ cells covering a square of size $1$ km$^2$ as shown in Fig.~\ref{FigSystemLayout}. The BSs are located at the four outer corners, to focus on the area that is jointly covered by the BSs. The users are uniformly and randomly distributed in the cell area, but the distance between a user and its serving BS is not less than $100$ m. There are $5$ users per cell: $K=5$. Each coherence interval has $200$ samples and a communication bandwidth of $20$ MHz is used. The noise variance is $-96$ dBm and the large-scale fading between user $k$ in cell $i$ and BS $l$ is modeled as
\begin{equation}
\beta_{i,k}^l = -128.1 - 37.6 \log b_{i,k}^l + z_{i,k}^l \; \mbox{(dB),}
\end{equation}
where $b_{i,k}^l$ measured in km is the distance between user $k$ in cell $i$ and BS $l$ and $z_{i,k}^l$ is the corresponding shadow fading, which is generated by the log-normal Gaussian distribution with standard derivation $7$ dB. For the double scattering model, we set the angular spread of BS antennas to $ \theta_{i,k}^l = 2 \pi/ 3$, which covers one sector in the cellular networks. The transmit scatterers are located at the distance of $0.2 b_{i,k}^l$ while the distance between  transmit and receive scatterers is $r_{i,k}^l = 0.7 b_{i,k}^l$. Moreover, the number of scatterers need to be large enough for small-scale fading to occur \cite{Gesbert2002a}; thus, we consider $S \in \{11, 21, 41\}$. The distance between two BS antennas is selected by using the carrier wavelength as a reference unit, $d_l \in \{ 0.1, 0.5, 1\}$. The number of BS antennas is $M = 100$. 

We assume that the cells use same orthogonal pilot sequences with the length $\tau_p = K$ (i.e., $f=1$). The users transmit at equal power and the median SNR at the cell edge is $-3$ dB. The channel models are normalized so that the average received power of a user is the same irrespective of the channel model. The performance is averaged over the $100$ random user locations and shadow fading realizations. For each set of user locations, the expectations in Theorem \ref{Theorem-Lower-Bound-Rate} are approximated using Monte-Carlo simulations with $1000$ realizations of the small-scale fading. 

\begin{figure}[t]
	\centering
	\includegraphics[trim=4cm 9cm 4.8cm 9cm, clip=true, width=2.9in]{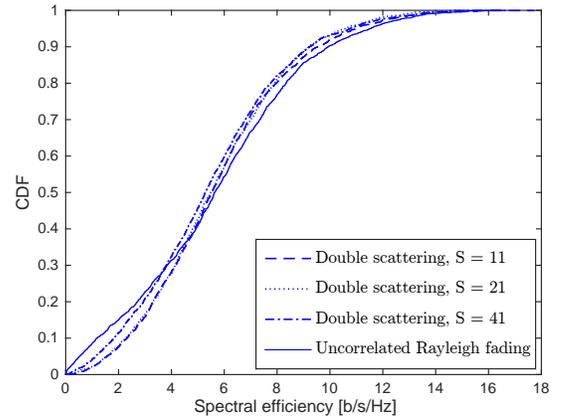} \vspace{-0.20cm}
	\caption{ Cumulative distribution function of UL spectral efficiency with MMSE detection and $d_l = 0.5$.}
	\label{FigCDFMMSEVariedScatters}
	\vspace{-0.40cm}
\end{figure}

Fig.~\ref{FigCDFMMSEVariedScatters} plots the cumulative distribution function (CDF) of the SE, using MMSE detection. The SE is very close to uncorrelated Rayleigh fading  even with a small number of scatterers (e.g., $S =11$). On average, the SE per user with $S=11$ is $5.65$ b/s/Hz while the system with uncorrelated Rayleigh fading gives $5.61$ b/s/Hz. Nonetheless, the difference between the two channel models can be more than $1$ b/s/Hz at the $95 \%$-likely SE point. Interestingly, in some realizations the double scattering model even provides better SE than uncorrelated Rayleigh fading. This is because the distributions of scatterers can make channels of two users more orthogonal than in uncorrelated Rayleigh fading.

\begin{figure}[t]
	\centering
	\includegraphics[trim=4cm 9cm 4.8cm 9cm, clip=true, width=2.9in]{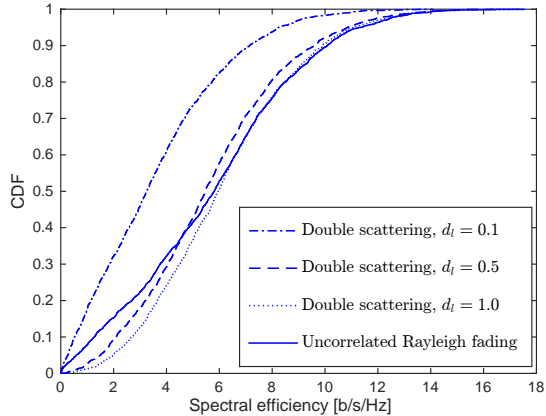} \vspace{-0.20cm} 
	\caption{Cumulative distribution function of UL spectral efficiency with different BS antenna spaces, $S= 21$, and MMSE detection.}
	\label{FigCDFVariedAntennaSpace}
		\vspace{-0.40cm}
\end{figure}

Fig.~\ref{FigCDFVariedAntennaSpace} shows the CDF of the SE for the systems with different BS antenna separations $d_l$. There exists a significant gap between $d_l = 0.1$
and the other curves ($d_l = 0.5$, $d_l = 1.0$, and uncorrelated Rayleigh). With $d_l =0.1$, on average the system can only provide SEs of about $3.11$ b/s/Hz,  while with $d_l = 0.5$ the average SE increases to $5.60$ b/s/Hz and is very similar to uncorrelated Rayleigh fading. Further antenna separation does not make any substantial difference. Consequently, we conclude that a BS antenna distance of half a carrier wavelength is enough to provide good performance.

Fig.~\ref{FigCDFVariedLinears} compares the SE using different linear detection techniques. ZF and MMSE are essentially equal and yield much higher performance that MR. For example, at $95 \%$-likely SE, ZF and MMSE can provide about $1.50$ b/s/Hz, but MR is only able to offer $0.13$ b/s/Hz. The equal performance of MMSE and ZF detection techniques is because $\mathbf{P}_l^{-1}$ is negligible compared to $(\widehat{\mathbf{H}}_l^l)^H \widehat{\mathbf{H}}_l^l$ when $M$ is large.

\begin{figure}[t]
	\centering
	\includegraphics[trim=4cm 9cm 4.8cm 9cm, clip=true, width=2.9in]{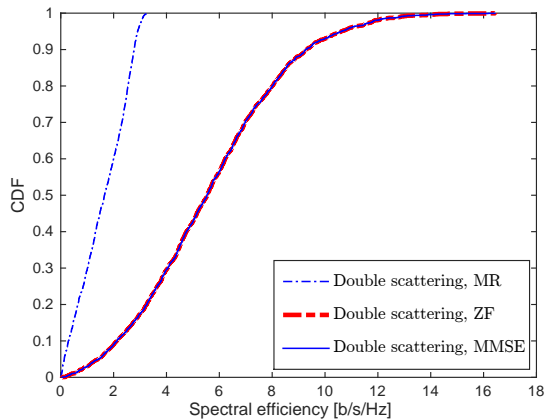} \vspace{-0.20cm}
	\caption{Cumulative distribution function of UL spectral efficiency with different linear detection techniques, $S= 21$, and $d_l = 0.5$.}
	\label{FigCDFVariedLinears}
	\vspace{-0.40cm}
\end{figure}

%=======================================================================
%=======================================================================
\section{Conclusion}
\vspace{-0.10cm}
The uncorrelated Rayleigh fading channel model has dominated the Massive MIMO literature, despite the fact that it is only justified in isotropic scattering environments. We have compared it with the double scattering channel model that may yield statistical properties far different from uncorrelated Rayleigh fading. The key differences of spatial correlation and favorable propagation behaviors were illustrated by simulations. Despite these important differences, the SE obtained with the two models are relatively similar even for a small number of scatterers, in particular if the antenna separation at the BS is at least half a carrier wavelength. Some users will obtain better SE with one model than the other, but the average performance is similar. Finally, in our simulations, ZF and MMSE yield similar performance while MR performs much worse.

\vspace{-0.30cm}


%\section*{Acknowledgment}
%This paper was supported by the European Union's Horizon 2020 research and innovation programme under grant agreement No 641985 (5Gwireless). It was also supported by ELLIIT and CENIIT.
%\vspace*{-0.2cm}



%\bibliographystyle{IEEEtran}

%==========================Reference==========================================================================
%\begin{small}
	\bibliographystyle{IEEEtran}
	\bibliography{IEEEabrv,refs}
%\end{small}

\end{document}